\documentclass[11pt]{article}
\usepackage{amssymb}
\usepackage[perpage,symbol]{footmisc}
\usepackage{listings}
\usepackage{amsfonts}
\usepackage{enumerate}
\usepackage{amsmath}
\usepackage{cite}
\usepackage{array}
\usepackage{booktabs}
\newcommand\relphantom[1]{\mathrel{\phantom{#1}}}

\begin{document}

\newtheorem{theorem}{Theorem}[section]
\newtheorem{corollary}[theorem]{Corollary}
\newtheorem{definition}[theorem]{Definition}
\newtheorem{proposition}[theorem]{Proposition}
\newtheorem{lemma}[theorem]{Lemma}
\newtheorem{example}[theorem]{Example}
\newenvironment{proof}{\noindent {\bf Proof.}}{\rule{3mm}{3mm}\par\medskip}
\newcommand{\remark}{\medskip\par\noindent {\bf Remark.~~}}

\title{A Class of Reducible Cyclic Codes and Their Weight Distribution\footnote{This work is supported by the
NSFC under Grants  11071160 and 11001170.}}
\author{Yan Liu\footnote{Corresponding author, Dept. of Math., SJTU, Shanghai, 200240,  liuyan0916@sjtu.edu.cn.}, Haode Yan\footnote{Dept. of Math., Shanghai Jiaotong Univ., Shanghai, 200240, hdyan@sjtu.edu.cn.}, Chunlei Liu\footnote{Dept. of Math., Shanghai Jiaotong Univ., Shanghai,
200240, clliu@sjtu.edu.cn.}}

\date{}
\maketitle
\thispagestyle{empty}

\abstract{In this paper, a family of reducible cyclic codes over $\mathbb{F}_{p}$  whose duals have four zeros is presented, where  $p$ is  an odd prime. Furthermore, the weight distribution of these cyclic codes is determined. }



\section{\small{INTRODUCTION}}
Throughout this paper, let $m \geq 5$ be an odd integer and $k$ be any positive integer  such that  gcd$(m,k)=1$. Let $p$ be  an odd prime and $\pi$ be a primitive element of the finite field $\mathbb{F}_{p^{m}}$.

Recall that an $[n,l,d]$ linear code $\mathcal{C}$ over $\mathbb{F}_{p}$ is a linear subspace of $\mathbb{F}_{p}^{n}$ with dimension $l$ and minimum Hamming distance $d$. Let $A_{i}$ denote the number of codewords  in $\mathcal{C}$ with Hamming weight $i$.
The sequence $(A_{0}, A_{1}, A_{2},\ldots, A_{n})$ is called the weight distribution of the code $\mathcal{C}$. $\mathcal{C}$ is called cyclic if for any $(c_{0}, c_{1},  \ldots, c_{n-1}) \in \mathcal{C}$, then $(c_{n-1}, c_{0},  \ldots, c_{n-2}) \in \mathcal{C}$.
A linear code $\mathcal{C}$ in $\mathbb{F}_{p}^{n}$ is cyclic if and only if $\mathcal{C}$ is an ideal of the polynomial residue class ring $\mathbb{F}_{p}[x]/(x^{n}-1)$. Since $\mathbb{F}_{p}[x]/(x^{n}-1)$ is a principal ideal ring, every cyclic code corresponds to a principal ideal $(g(x))$ of  the multiples of a polynomial $g(x)$ which is the monic polynomial of lowest degree in the ideal. This polynomial  $g(x)$ is called the generator polynomial, and $h(x)=(x^{n}-1)/g(x)$ is called the parity-check polynomial of the code $\mathcal{C}$. We also  recall that a cyclic code is called irreducible if its parity-check polynomial  is irreducible over $\mathbb{F}_{p}$ and reducible, otherwise.

 Determining the weight distribution of a linear code is an important research object in coding theory.
 For cyclic codes, the error-correcting capability may not be as good as with some other linear codes in general. However,  because of  their good algebraic
structure, the weight distribution of some cyclic codes can be determined by algebraic techniques, exponential sums for example. Besides, cyclic codes have wide applications in storage and communication systems because they have efficient encoding and decoding algorithms. Therefore, the weight distributions of cyclic codes is not only a problem of theoretical interest, but also of practical importance.
For information on the weight distribution of  cyclic codes, the reader is referred to \cite{1,2,3,4,5,6,7,8,9,10,18,12,13,14,15,19,20,21,22}.

Let $h_{0}(x)$, $h_{1}(x)$, $h_{2}(x)$ and $h_{3}(x)$ be the minimal polynomials of $\pi^{-1}$, $\pi^{-2}$, $\pi^{-(p^{k}+1)}$ and $\pi^{-(p^{2k}+1)}$ over $\mathbb{F}_{p}$, respectively. It is easy to check that $h_{0}(x)$, $h_{1}(x)$ $h_{2}(x)$ and $h_{3}(x)$ are polynomials of degree $m$ and are pairwise distinct when $m\geq 5$. Define $h(x)=h_{0}(x)h_{1}(x)h_{2}(x)h_{3}(x)$. Then $h(x)$ has degree $4m$ and is a factor of $x^{p^{m}-1}-1$.

Let $\mathcal{C}_{(p,m,k)}$ be the cyclic code with parity-check polynomial $h(x)$. Then $\mathcal{C}_{(p,m,k)}$ has length $p^{m}-1$ and dimension $4m$. Moreover, it can be expressed as \[\mathcal{C}_{(p,m,k)}=\{\mathbf{c}_{(\alpha,\beta,\gamma,\delta)}: \alpha,\beta,\gamma,\delta \in \mathbb{F}_{p^{m}}\},\] where \[\mathbf{c}_{(\alpha,\beta,\gamma,\delta)}=\big(Tr(\alpha\pi^{(p^{2k}+1)t}+\beta\pi^{(p^{k}+1)t}+\gamma\pi^{2t}+\delta\pi^{t})\big)_{t=0}^{p^{m}-2}\] and $Tr$ is the trace map from $\mathbb{F}_{p^{m}}$ to $\mathbb{F}_{p}$.
Let $h'(x)=h_{1}(x)h_{2}(x)h_{3}(x)$ and $\mathcal{C'}$  be the cyclic code with parity-check polynomial $h'(x)$. Then $\mathcal{C'}$ is a subcode of $\mathcal{C}_{(p,m,k)}$ with dimension $3m$. From \cite{22}  and \cite{20}, we can obtain the weight distribution of  $\mathcal{C}'_{(p,m,k)}$  when $m$ is odd. In this paper, we will determine the weight distribution of $\mathcal{C}_{(p,m,k)}$. For doing this, we need to determine the value distribution of the  multi-sets
\begin{equation}\label{Eq:2.1}
\{T(\alpha,\beta,\gamma)=\sum_{x \in  \mathbb{F}_{p^{m}}}\zeta_{p}^{Tr(\alpha x^{p^{2k}+1}+\beta x^{p^{k}+1}+\gamma x^{2})}: \alpha,\beta,\gamma \in  \mathbb{F}_{p}^{m}\}
\end{equation}
and
\begin{equation}\label{Eq:2.2}
\{S(\alpha,\beta,\gamma,\delta)=\sum_{x \in  \mathbb{F}_{p^{m}}}\zeta_{p}^{Tr(\alpha x^{p^{2k}+1}+\beta x^{p^{k}+1}+\gamma x^{2}+\delta x)}:\alpha,\beta,\gamma,\delta \in  \mathbb{F}_{p}^{m}\},
\end{equation}
where $\zeta_{p}=e^{\frac{2\pi i}{p}}$.
The rest of this paper is  organized as follows. Some preliminaries   will be introduced in Section 2. The weight  distribution of the cyclic code $\mathcal{C}_{(p,m,k)}$ will be given in Section 3.

\section{\small{MATHEMATICAL FOUNDATIONS}}
In this section, we first give a brief introduction to the theory of quadratic forms over finite fields.
By  fixing a basis $v_{1}, v_{2},\ldots, v_{m}$ of $\mathbb{F}_{p^{m}}$ over $\mathbb{F}_{p}$, each $x \in \mathbb{F}_{p^{m}}$ can be uniquely expressed as
  \[
  x=x_{1}v_{1}+x_{2}v_{2}+\cdots+x_{m}v_{m},
  \] where $x_{i} \in \mathbb{F}_{p}$ for $1\leq i \leq m$. Then $\mathbb{F}_{p^{m}}$ is isomorphic to the $m$-dimensional linear space $\mathbb{F}_{p}^{m}$.
In other words, we have the following $\mathbb{F}_{p}$-linear isomorphism:
  \[
  x=x_{1}v_{1}+x_{2}v_{2}+\cdots+x_{m}v_{m} \mapsto X=(x_{1}, x_{2},\ldots, x_{m}).
  \]
For a quadratic form  $F$, there exists a symmetric matrix $H$ of order $m$ over $\mathbb{F}_{p}$ such that $F(X)= XHX'$, where $X=(x_{1}, x_{2},\ldots,x_{m})\in \mathbb{F}_{p}^{m}$  and $X'$ denotes the transpose of $X$.
The rank of the quadratic form $F$ is defined as the codimension of the $\mathbb{F}_{p}$-vector space $V=\{x \in \mathbb{F}_{p^{m}}: f(x+z)-f(x)-f(z)=0~ for~ all ~z \in \mathbb{F}_{p^{m}} \}$, i.e., the rank of $H$.
 Then there exists a nonsingular matrix $M$ of order $m$ over $\mathbb{F}_{p}$ such that $MHM'$ is a diagonal matrix (\cite{11}). Under the nonsingular linear substitution $X=ZM$ with $Z=(z_{1}, z_{2},\ldots,z_{m})\in \mathbb{F}_{p}^{m}$, then $F(X)=ZMHM'Z'=\sum_{i=1}^{r}d_{i}z_{i}^{2}$, where $r$ is the rank of $F(X)$ and $d_{i} \in \mathbb{F}_{p}^{\ast}$.   We can recall that the Legendre symbol $(\frac{a}{p})$ has the value 1 if $a$ is a quadratic residue mod $p$, $-1$ if $a$ is a quadratic nonresidue mod $p$, and zero if $p | a$. Let $\triangle= d_{1}d_{2}\cdots d_{r}$ (we assume $\triangle=0$ when $r=0$). Then the $(\frac{\triangle}{p})$ is an invariant of $H$ under the action of $M \in GL_{m}(\mathbb{F}_{p})$.

For any fixed $(\alpha,\beta,\gamma) \in \mathbb{F}_{p^m}^{3}$,  let $Q_{\alpha,\beta,\gamma}(x)=Tr(\alpha x^{p^{2k}+1}+\beta x^{p^{k}+1}+\gamma x^{2} )$. Then its induced quadratic form is
\[\begin{split}
F_{\alpha,\beta,\gamma}(X)&=Tr(\alpha(\sum_{i=1}^{m}x_{i}v_{i})^{p^{2k}+1}+\beta(\sum_{i=1}^{m}x_{i}v_{i})^{p^{k}+1}+\gamma
      (\sum_{i=1}^{m}x_{i}v_{i})^{2})\\
    &=\sum_{i,j=1}^{m}Tr(\alpha v_{i}^{p^{2k}+1}v_{j}+\beta v_{i}^{p^{k}+1}v_{j}+\gamma v_{i}v_{j})x_{i}x_{j}\\
    &=XH_{\alpha,\beta,\gamma}X',
\end{split}
\]
where $H_{\alpha,\beta,\gamma}=(h_{i,j})$ and
\[h_{i,j}=\frac{1}{2}Tr(\alpha (v_{i}^{p^{2k}+1}v_{j}+v_{i}v_{j}^{p^{2k}+1})+\beta (v_{i}^{p^{k}+1}v_{j}+v_{i}v_{j}^{p^{k}+1}))+\gamma v_{i}v_{j},~~1\leq i, j\leq m.
\]
Then we have the following:
\begin{equation*}
T(\alpha,\beta,\gamma)=\sum_{x \in  \mathbb{F}_{p^{m}}}\zeta_{p}^{Tr(\alpha x^{p^{2k}+1}+\beta x^{p^{k}+1}+\gamma x^{2})}=\sum_{X \in  \mathbb{F}_{p}^{m}}\zeta_{p}^{XH_{\alpha,\beta,\gamma}X'}
\end{equation*}
and
\begin{equation*}
S(\alpha,\beta,\gamma,\delta)=\sum_{x \in  \mathbb{F}_{p^{m}}}\zeta_{p}^{Tr(\alpha x^{p^{2k}+1}+\beta x^{p^{k}+1}+\gamma x^{2}+\delta x)}=\sum_{X \in  \mathbb{F}_{p}^{m}}\zeta_{p}^{XH_{\alpha,\beta,\gamma}X'+A_{\delta}X'},
\end{equation*} where $A_{\delta}=(Tr(\delta v_{_{1}}), Tr(\delta v_{_{2}}),\ldots, Tr(\delta v_{_{m}}))$.
Hence, in order to determine the  value distribution of the two multi-sets
\[
\{T(\alpha,\beta,\gamma)=\sum_{x \in  \mathbb{F}_{p^{m}}}\zeta_{p}^{Tr(\alpha x^{p^{2k}+1}+\beta x^{p^{k}+1}+\gamma x^{2})}: \alpha,\beta,\gamma \in  \mathbb{F}_{p}^{m}\}
\]
and
\[
\{S(\alpha,\beta,\gamma,\delta)=\sum_{x \in  \mathbb{F}_{p^{m}}}\zeta_{p}^{Tr(\alpha x^{p^{2k}+1}+\beta x^{p^{k}+1}+\gamma x^{2}+\delta x)}:\alpha,\beta,\gamma,\delta \in  \mathbb{F}_{p}^{m}\},
\]
we need the following lemmas.
\begin{lemma}[\cite{9}]\label{Le:2.1}
Let $F(X)=XHX'$ be a quadratic form in $m$ variables of rank $r$ over $\mathbb{F}_{p}$, then
 \begin{enumerate}
  \item  $\sum_{X \in  \mathbb{F}_{p}^{m}}\zeta_{p}^{F(X)}=
           (\frac{\triangle}{p})(p^{\ast})^{\frac{r}{2}}p^{m-r}$, where $p^{\ast}=(-1)^{\frac{p-1}{2}}p$;\\
  \item for $A=(a_{1}, a_{2}, \ldots, a_{m}) \in \mathbb{F}_{p}^{m}$, if $2YH+A=0$ has a solution $Y=B\in \mathbb{F}_{p}^{m}$, then $\sum_{X \in \mathbb{F}_{p}^{m}}\zeta_{p}^{F(X)+AX'}=\zeta_{p}^{t}\sum_{X \in \in \mathbb{F}_{p}^{m}}\zeta_{p}^{F(X)}$, where $t=\frac{1}{2}AB' \in \mathbb{F}_{p}$, otherwise $\sum_{X \in \in \mathbb{F}_{p}^{m}}\zeta_{p}^{F(X)+AX'}=0$.
 \end{enumerate}

\end{lemma}
\begin{lemma}[\cite{20,22}]\label{Le:2.2}
\begin{enumerate}
\item For any $(\alpha,\beta,\gamma) \in \mathbb{F}_{p^{m}}^{3}\backslash \{(0,0,0)\}$,  $rank (H_{\alpha,\beta,\gamma})$ is at least $m-4$.\\
\item Let $n_{i}$ be the number of $(\alpha,\beta,\gamma)$ such that $r_{\alpha,\beta,\gamma}=m-i$ for $0\leq i \leq 4$. Then
    \[
    n_{1}=\frac{p^{m+1}(p^{2m}-p^{2m-2}-p^{2m-3}+p^{m-2}+p^{m-3}-1)}{p^{2}-1}
    \] and \[
    n_{3}=\frac{p^{m-3}(p^{m-1}-1)(p^{m}-1)}{p^{2}-1}.
    \]
\end{enumerate}
\end{lemma}
The following result is important to determine the multiplicity of each value  of $T(\alpha,\beta,\gamma)$.
\begin{lemma}\label{Le:2.3}
For the exponential sum $T(\alpha,\beta,\gamma)$, we have
\[\sum_{\alpha,\beta,\gamma \in \mathbb{F}_{p^m}}T^{2}(\alpha,\beta,\gamma)=\begin{cases}
         p^{3m}(2p^{m}-1),& q \equiv 1\pmod4,\\
         p^{3m},& q \equiv 3\pmod4,
        \end{cases}\]
        and
 \[\sum_{\alpha,\beta,\gamma \in \mathbb{F}_{p^m}}T^{4}(\alpha,\beta,\gamma)=\begin{cases}
p^{3m}((2p^{m}-1)^{2}+(p-1)(p^{m}-1)(2p^{m}-p-1)),& q \equiv 1\pmod4,\\
p^{3m}(1+(p+1)(p^{m}-1)(2p^{m}-p+1)),& q \equiv 3\pmod4.
\end{cases}\]
\end{lemma}
We prove this lemma only for the case that $p\equiv 3 \pmod 4$. The proof for the case that $p\equiv 1 \pmod 4$ is similar and omitted. To prove this, the following two lemmas are necessary.
\begin{lemma}\label{Le:2.4}
Let $p \equiv 3 \pmod 4$ and let $N_{2}$ denote the number of solutions $(x_{1},x_{2}) \in \mathbb{F}_{p^{m}}^{2}$ of the following system of equations
\begin{equation*}
\begin{cases}
x_{1}^{2}+ x_{2}^{2}=0\\
x_{1}^{p^{k}+1}+x_{2}^{p^{k}+1}=0\\
x_{1}^{p^{2k}+1}+x_{2}^{p^{2k}+1}=0.
\end{cases}
\end{equation*}
Then $N_{2}=1$.
\end{lemma}
\begin{proof}
This system of equations have only one solution $(0,0)$, since  $-1$ is a  non-square when $p \equiv 3 \pmod 4$.
\end{proof}
\begin{lemma}\label{Le:2.5}
Let $p \equiv 3 \pmod 4$ and let $N_{4}$ denote the number of solutions $(x_{1},x_{2},x_{3},x_{4}) \in \mathbb{F}_{p^{m}}^{4}$ of the following system of equations
\[
\begin{cases}
x_{1}^{2}+x_{2}^{2}+ x_{3}^{2}+ x_{4}^{2}=0\\
x_{1}^{p^{k}+1}+x_{2}^{p^{k}+1}+ x_{3}^{p^{k}+1}+ x_{4}^{p^{k}+1}=0\\
x_{1}^{p^{2k}+1}+x_{2}^{p^{2k}+1}+ x_{3}^{p^{2k}+1}+ x_{4}^{p^{2k}+1}=0.
\end{cases}
\]
Then $N_{4}=(p+1)(p^{m}-1)(2p^{m}-p+1)+1$.
\end{lemma}
\begin{proof}
See Appendix.
\end{proof}
Now we are ready to prove Lemma \ref{Le:2.3} in the case of $p \equiv 3 \pmod 4$.

\noindent $\mathbf{Proof}$ $\mathbf{of}$ $\mathbf{Lemma}$ \ref{Le:2.3}. From Eq. (\ref{Eq:2.1}), we have

\begin{equation*}
\begin{split}
&\sum_{(\alpha,\beta,\gamma) \in \mathbb{F}_{p^{m}}^{3}}T^{2}(\alpha,\beta,\gamma)\\
&=\sum_{(\alpha,\beta,\gamma) \in \mathbb{F}_{p^{m}}^{3}}\sum_{x_{1} \in  \mathbb{F}_{p^{m}}}\zeta_{p}^{Tr(\alpha x_{1}^{p^{2k}+1}+\beta x_{1}^{p^{k}+1}+\gamma x_{1}^{2})}\sum_{x_{2} \in  \mathbb{F}_{p^{m}}}\zeta_{p}^{Tr(\alpha x_{2}^{p^{2k}+1}+\beta x_{2}^{p^{k}+1}+\gamma x_{2}^{2})}\\
&=\sum_{(x_{1},x_{2}) \in \mathbb{F}_{p^{m}}^{2}}\sum_{\alpha \in \mathbb{F}_{p^{m}}}\zeta_{p}^{Tr(\alpha(x_{1}^{p^{2k}+1}+x_{2}^{p^{2k}+1}))}\sum_{\beta \in \mathbb{F}_{p^{m}}}\zeta_{p}^{Tr(\beta(x_{1}^{p^{k}+1}+x_{2}^{p^{k}+1}))}\sum_{\gamma \in \mathbb{F}_{p^{m}}}\zeta_{p}^{Tr(\gamma(x_{1}^{2}+x_{2}^{2}))}\\
&=p^{3m}\#W,
\end{split}
\end{equation*}
where
\begin{equation*}
W=\{(x_{1},x_{2}) \in \mathbb{F}_{p^{m}}^{2}: x_{1}^{2}+x_{2}^{2}=0, x_{1}^{p^{k}+1}+x_{2}^{p^{k}+1}=0,x_{1}^{p^{2k}+1}+x_{2}^{p^{2k}+1}=0\}.
\end{equation*}
Then by Lemma \ref{Le:2.4}, we have
\begin{equation}\label{Eq:2.3}
\sum_{(\alpha,\beta,\gamma) \in \mathbb{F}_{p^{m}}^{3}}T^{2}(\alpha,\beta,\gamma)=p^{3m}.
\end{equation}
Similarly, by Lemmas \ref{Le:2.5}, we have
\begin{equation}\label{Eq:2.4}
\sum_{(\alpha,\beta,\gamma) \in \mathbb{F}_{p^{m}}^{3}}T^{4}(\alpha,\beta,\gamma)=p^{3m}(1+(p+1)(p^{m}-1)(2p^{m}-p+1)).
\end{equation}

\section{\small{RESULTS ON EXPONENTIAL SUMS AND THE WEIGHT DISTRIBUTION OF THE CYCLIC CODE}}
We follow the notation and conditions fixed in Section 1 and 2.
\begin{theorem}\label{Th:3.1}
Let $m \geq 5$ be an odd integer, $k$ be any positive integer  such that  gcd$(m,k)=1$. Then the value distribution of the multi-set $\{T(\alpha,\beta,\gamma)\mid \alpha,\beta,\gamma \in \mathbb{F}_{p^{m}}\}$ is given by Table 1.
\end{theorem}
\begin{proof}
By Lemma \ref{Le:2.1}, for $\varepsilon =\pm 1$ and $0\leq i\leq 4$, we define
\[
n_{\varepsilon,i}=\begin{cases}
\#\{(\alpha,\beta,\gamma) \in \mathbb{F}_{p^{m}}^{3}\backslash \{(0,0,0)\} \mid T(\alpha,\beta,\gamma)=\varepsilon p^{\frac{m+i}{2}}\}& m-i~even,\\
\#\{(\alpha,\beta,\gamma) \in \mathbb{F}_{p^{m}}^{3}\backslash \{(0,0,0)\} \mid T(\alpha,\beta,\gamma)=\varepsilon \sqrt{p^{\ast}}p^{\frac{m+i-1}{2}}\}& m-i~odd.
\end{cases}
\]
\begin{itemize}
  \item $i= 1, 3$. In this case, $m-i$ is even. According to the results of \cite{22} and \cite{20}, we can obtain
  \[
  \begin{split}
  &n_{1,1}=\frac{(p^{m+1}+p^{(m+3)/2})(p^{2m}-p^{2m-2}-p^{2m-3}+p^{m-2}+p^{m-3}-1)}{2(p^{2}-1)},\\
  &n_{-1,1}=\frac{(p^{m+1}-p^{(m+3)/2})(p^{2m}-p^{2m-2}-p^{2m-3}+p^{m-2}+p^{m-3}-1)}{2(p^{2}-1)},\\
  &n_{1,3}=\frac{(p^{m-3}+p^{(m-3)/2})(p^{m-1}-1)(p^{m}-1)}{2(p^{2}-1)},\\
  &n_{-1,3}=\frac{(p^{m-3}-p^{(m-3)/2})(p^{m-1}-1)(p^{m}-1)}{2(p^{2}-1)}.
  \end{split}
  \]
  \item $i= 0,2,4$. In this case, $m-i$ is odd. By the same method in \cite{9}, we also have
  \[
  n_{1,i}=n_{-1,i}=\frac{1}{2}n_{i}.
  \]
Moreover, we have
\begin{equation*}
\begin{split}
&\sum_{(\alpha,\beta,\gamma) \in \mathbb{F}_{p^{m}}^{3}}T^{2}(\alpha,\beta,\gamma)\\
&=p^{2m}+(n_{1,0}+n_{-1,0})p^{\ast}p^{m-1}+(n_{1,1}+n_{-1,1})p^{m+1}+(n_{1,2}+n_{-1,2})p^{\ast}p^{m+1}\\
 &\relphantom{=} {}+(n_{1,3}+n_{-1,3})p^{m+3}+(n_{1,4}+n_{-1,4})p^{\ast}p^{m+3}\\
 &=p^{2m}+2n_{1,0}p^{\ast}p^{m-1}+n_{1}p^{m+1}+2n_{1,2}p^{\ast}p^{m+1}
+n_{3}p^{m+3}+2n_{1,4}p^{\ast}p^{m+3}
\end{split}
\end{equation*}
and similarly,
\begin{equation*}
\begin{split}
&\sum_{(\alpha,\beta,\gamma) \in \mathbb{F}_{p^{m}}^{3}}T^{4}(\alpha,\beta,\gamma)\\
&=p^{4m}+2n_{1,0}p^{2m}+n_{1}p^{2(m+1)}+2n_{1,2}p^{2(m+2)}
+n_{3}p^{2(m+3)}+2n_{1,4}p^{2(m+4)}.
\end{split}
\end{equation*}
By the two equations above, Lemma \ref{Le:2.2} and \ref{Le:2.3}, we can obtain the following system of equations:
\begin{equation*}
\begin{cases}
2n_{1,0}+2n_{1,2}+2n_{1,4}=(p^{m}-1)(p^{2m}-p^{2m-1}+p^{2m-4}+p^{m}-p^{m-1}-p^{m-3}+1) \\
2n_{1,0}+2n_{1,2}p^{2}+2n_{1,4}p^{4}=p^{3m}-p^{2m}\\
2n_{1,0}+2n_{1,2}p^{4}+2n_{1,4}p^{8}=p^{3m+1}+p^{3m}-p^{2m+2}-p^{2m+1}-(p^{2m}+p^{m+2}.\\
\end{cases}
\end{equation*}
By solving this system of equations, we can have
\[
\begin{split}
&n_{1,0}=n_{-1,0}=\frac{1}{2}(p^{m}-1)(p^{2m}-p^{2m-1}+p^{2m-4}+p^{m}-p^{m-1}-p^{m-3}+1)\\
  &\relphantom{=} {}-n_{1,2}-n_{1,4},\\
&n_{1,2}=n_{-1,2}=\frac{p^{2}(p^{m-1}-1)(p^{2m}-p^{2m-2}-p^{2m-3}+p^{m-2}+p^{m-3}-1)}{2(p^{2}-1)^{2}},\\
&n_{1,4}=n_{-1,4}=\frac{(p^{m}-1)(p^{m-1}-1)(p^{m-3}-1)}{2(p^{2}-1)(p^{4}-1)}.
\end{split}
\]
\end{itemize}
The proof is completed.
\end{proof}
\begin{table}[htbp]
\caption{Value Distribution of $T(\alpha,\beta,\gamma)$}
\centering
\begin{tabular}{ll}
 \hline
 Value& Frequency\\
 \hline
 $\sqrt{p^{\ast}}p^{\frac{m-1}{2}}, -\sqrt{p^{\ast}}p^{\frac{m-1}{2}}$ & $n_{1,0}(n_{-1,0})$\\
 $p^{\frac{m+1}{2}}$ & $n_{1,1}$\\
 $-p^{\frac{m+1}{2}}$ & $n_{-1,1}$\\
 $\sqrt{p^{\ast}}p^{\frac{m+1}{2}}, -\sqrt{p^{\ast}}p^{\frac{m+1}{2}}$ & $n_{1,2}(n_{-1,2})$\\
 $p^{\frac{m+3}{2}}$ & $n_{1,3}$\\
 $-p^{\frac{m+3}{2}}$ & $n_{-1,3}$\\
  $\sqrt{p^{\ast}}p^{\frac{m+3}{2}}, -\sqrt{p^{\ast}}p^{\frac{m+3}{2}}$ & $n_{1,4}(n_{-1,4})$\\
 $p^{m}$ & 1\\
 \hline
\end{tabular}
\end{table}
Until now, we have determined the value distribution of the multi-set $\{T(\alpha,\beta,\gamma)\mid \alpha,\beta,\gamma \in \mathbb{F}_{p^{m}}\}$. The value distribution of the multi-set $\{S(\alpha,\beta,\gamma,\delta)\mid \alpha,\beta,\gamma,\delta \in \mathbb{F}_{p^{m}}\}$ can be determined by the following theorem.
\begin{theorem}\label{Th:3.2}
Let $m \geq 5$ be an odd integer, $k$ be any positive integer  such that  gcd$(m,k)=1$. Then the value distribution of the multi-set $\{S(\alpha,\beta,\gamma,\delta)\mid \alpha,\beta,\gamma,\delta \in \mathbb{F}_{p^{m}}\}$ is given by Table 2.
\end{theorem}
\begin{proof}
By Lemma \ref{Le:2.1}, $S(\alpha,\beta,\gamma,\delta)$  takes values from the set
\[
\{0, \pm \zeta_{p}^{j} p^{\frac{m+i}{2}}, \pm \zeta_{p}^{j} \sqrt{p^{\ast}}p^{\frac{m+i-1}{2}}: j \in \mathbb{F}_{p}, 0\leq i\leq 4\}.
 \] Then for $\varepsilon=\pm1$, we define
\[
n_{\varepsilon,i,j}=\begin{cases}
\#\{(\alpha,\beta,\gamma,\delta) \in \mathbb{F}_{p^{m}}^{3}\backslash \{(0,0,0,0)\} \mid S(\alpha,\beta,\gamma,\delta)=\varepsilon \zeta_{p}^{j} p^{\frac{m+i}{2}}\}& m-i~even,\\
\#\{(\alpha,\beta,\gamma,\delta) \in \mathbb{F}_{p^{m}}^{3}\backslash \{(0,0,0,0)\} \mid S(\alpha,\beta,\gamma,\delta)=\varepsilon \zeta_{p}^{j} \sqrt{p^{\ast}}p^{\frac{m+i-1}{2}}\}& m-i~odd.
\end{cases}
\]
and
\[\omega=\#\{(\alpha,\beta,\gamma,\delta) \in \mathbb{F}_{p^{m}}^{3}\backslash \{(0,0,0,0)\} \mid S(\alpha,\beta,\gamma,\delta)=0\}.\]
By the same method in \cite{9}, we have the following results.
\begin{itemize}
  \item $i= 1, 3$. In this case, $m-i$ is even. For $\varepsilon=\pm 1$ and $j \in \mathbb{F}_{p}^{\ast}$, we have
  \[
  \begin{split}
  &n_{\varepsilon,i,0}=(p^{m-i-1}+\varepsilon(p-1)p^{\frac{m-i-2}{2}})n_{\varepsilon,i},\\
  &n_{\varepsilon,i,j}=(p^{m-i-1}-\varepsilon p^{\frac{m-i-2}{2}})n_{\varepsilon,i}.
  \end{split}
  \]
  \item $i= 0,2,4$. In this case, $m-i$ is odd.  For $\varepsilon=\pm 1$ and $j \in \mathbb{F}_{p}^{\ast}$, we also have
  \[
  \begin{split}
  &n_{\varepsilon,i,0}=p^{m-i-1}n_{\varepsilon,i},\\
  &n_{\varepsilon,i,j}=(p^{m-i-1}+\varepsilon(\frac{-j}{p}) p^{\frac{m-i-1}{2}})n_{\varepsilon,i}.
  \end{split}
  \]
Furthermore, we have
\[
\omega=p^{m}-1+(p^{m}-p^{m-1})n_{1}+(p^{m}-p^{m-2})n_{2}+(p^{m}-p^{m-3})n_{3}+(p^{m}-p^{m-4})n_{4}.
\]
\end{itemize}
Summarizing the discussion above completes the proof of this theorem.
\end{proof}
\begin{table}[htbp]
\caption{Value Distribution of $S(\alpha,\beta,\gamma,\delta)$}
\centering
\begin{tabular}{ll}
 \hline
 Value& Frequency\\
 \hline
 $\sqrt{p^{\ast}}p^{\frac{m-1}{2}}, -\sqrt{p^{\ast}}p^{\frac{m-1}{2}}$ & $n_{1,0,0}(n_{-1,0,0})$\\
 $\zeta_{p}^{j}\sqrt{p^{\ast}}p^{\frac{m-1}{2}}, j \in \mathbb{F}_{p}^{\ast}$ & $n_{1,0,j}$\\
 $-\zeta_{p}^{j}\sqrt{p^{\ast}}p^{\frac{m-1}{2}}, j \in \mathbb{F}_{p}^{\ast}$ & $n_{-1,0,j}$\\
 $p^{\frac{m+1}{2}}$ & $n_{1,1,0}$\\
 $-p^{\frac{m+1}{2}}$ & $n_{-1,1,0}$\\
  $\zeta_{p}^{j}p^{\frac{m+1}{2}}, j \in \mathbb{F}_{p}^{\ast}$ & $n_{1,1,j}$\\
  $-\zeta_{p}^{j}p^{\frac{m+1}{2}}, j \in \mathbb{F}_{p}^{\ast}$ & $n_{-1,1,j}$\\
 $\sqrt{p^{\ast}}p^{\frac{m+1}{2}}, -\sqrt{p^{\ast}}p^{\frac{m+1}{2}}$ & $n_{1,2,0}(n_{-1,2,0})$\\
 $\zeta_{p}^{j}\sqrt{p^{\ast}}p^{\frac{m+1}{2}}, j \in \mathbb{F}_{p}^{\ast}$ & $n_{1,2,j}$\\
 $-\zeta_{p}^{j}\sqrt{p^{\ast}}p^{\frac{m+1}{2}}, j \in \mathbb{F}_{p}^{\ast}$ & $n_{-1,2,j}$\\
 $p^{\frac{m+3}{2}}$ & $n_{1,3,0}$\\
 $-p^{\frac{m+3}{2}}$ & $n_{-1,3,0}$\\
 $\zeta_{p}^{j}p^{\frac{m+3}{2}}, j \in \mathbb{F}_{p}^{\ast}$ & $n_{1,3,j}$\\
  $-\zeta_{p}^{j}p^{\frac{m+3}{2}}, j \in \mathbb{F}_{p}^{\ast}$ & $n_{-1,3,j}$\\
  $\sqrt{p^{\ast}}p^{\frac{m+3}{2}}, -\sqrt{p^{\ast}}p^{\frac{m+3}{2}}$ & $n_{1,4,0}(n_{-1,4,0})$\\
  $\zeta_{p}^{j}\sqrt{p^{\ast}}p^{\frac{m+3}{2}}, j \in \mathbb{F}_{p}^{\ast}$ & $n_{1,4,j}$\\
  $-\zeta_{p}^{j}\sqrt{p^{\ast}}p^{\frac{m+3}{2}}, j \in \mathbb{F}_{p}^{\ast}$ & $n_{-1,4,j}$\\
  0& $\omega$\\
 $p^{m}$ & 1\\
 \hline
\end{tabular}
\end{table}

\begin{theorem}\label{Th:3.11}
Let $m\geq 5$ be an odd integer and $k$ be any positive integer such that gcd$(m,k)=1$.
$\mathcal{C}_{(p,m,k)}$ is a  cyclic code over $\mathbb{F}_{p}$ with parameters $[p^{m}-1, 4m, (p-1)(p^{m-1}-p^{\frac{m+3}{2}})]$. Moreover, the weight distribution of $\mathcal{C}_{(p,m,k)}$ is given in Table 3.
\end{theorem}
\begin{proof}
According to the discussion in  section 1,
the length and dimension of $\mathcal{C}_{(p,m,k)}$ are clearly.
In terms of exponential sums, the weight of the codeword $\mathbf{c}_{(\alpha,\beta,\gamma,\delta)}=(c_{0}, c_{1},\ldots,c_{p^{m}-2})$ in $\mathcal{C}_{(p,m,k)}$ is given by
\begin{equation*}
\begin{split}
 W(\mathbf{c}_{(\alpha,\beta,\gamma,\delta)})
 &= \#\{0\leq t\leq p^{m}-2: c_{t}\neq 0\}\\
  &= p^{m}-1-\frac{1}{p}\sum_{t=0}^{p^{m}-2}\sum_{y \in \mathbb{F}_{p}}\zeta_{p}^{yc(t)} \\
   &=  p^{m}-1-\frac{1}{p}\sum_{t=0}^{p^{m}-2}\sum_{y \in \mathbb{F}_{p}}\zeta_{p}^{yTr(\alpha\pi^{(p^{2k}+1)t}+\beta\pi^{(p^{k}+1)t}+\gamma\pi^{2t}+\delta\pi^{t})}\\
      &=   p^{m}-1- \frac{1}{p}\sum_{y \in \mathbb{F}_{p}}\sum_{x \in \mathbb{F}_{p^{m}}^{\ast}}\zeta_{p}^{yTr(\alpha x^{p^{2k}+1}+\beta x^{p^{k}+1}+\gamma x^{2}+\delta x)}\\
   &=   p^{m-1}(p-1)-\frac{1}{p}\sum_{y \in \mathbb{F}_{p}^{\ast}}\sum_{x \in \mathbb{F}_{p^{m}}}\zeta_{p}^{yTr(\alpha x^{p^{2k}+1}+\beta x^{p^{k}+1}+\gamma x^{2}+\delta x)}\\
   &=   p^{m-1}(p-1)-\frac{1}{p}R(\alpha,\beta,\gamma,\delta),
\end{split}
\end{equation*}
where
\[
R(\alpha,\beta,\gamma,\delta)=\sum_{y \in \mathbb{F}_{p}^{\ast}}S(y\alpha,y\beta,y\gamma,y\delta).
 \]
 Again by the same method in \cite{9}, we have the following results.
\begin{itemize}
  \item $i= 1, 3$. In this case, $m-i$ is even. For $\varepsilon=\pm 1$ and $j \in \mathbb{F}_{p}^{\ast}$, we have
  \begin{itemize}
     \item  if $S(\alpha,\beta,\gamma,\delta)=\varepsilon p^{\frac{m+i}{2}}$, then
             \[
                  R(\alpha,\beta,\gamma,\delta)=\varepsilon (p-1) p^{\frac{m+i}{2}};
              \]
     \item if $S(\alpha,\beta,\gamma,\delta)=\varepsilon \zeta_{p}^{j} p^{\frac{m+i}{2}}$, then
              \[
                R(\alpha,\beta,\gamma,\delta)=-\varepsilon p^{\frac{m+i}{2}}.
              \]
   \end{itemize}
  \item $i= 0,2,4$. In this case, $m-i$ is odd.  For $\varepsilon=\pm 1$ and $j \in \mathbb{F}_{p}^{\ast}$, we also have
    \begin{itemize}
     \item  if $S(\alpha,\beta,\gamma,\delta)=\varepsilon\sqrt{p^{\ast}}  p^{\frac{m+i-1}{2}}$, then
             \[
                  R(\alpha,\beta,\gamma,\delta)=0;
              \]
     \item if $S(\alpha,\beta,\gamma,\delta)=\varepsilon \zeta_{p}^{j}\sqrt{p^{\ast}} p^{\frac{m+i-1}{2}}$, then
              \[
                R(\alpha,\beta,\gamma,\delta)=\varepsilon (\frac{-j}{p})p^{\frac{m+i+1}{2}}.
              \]
   \end{itemize}
 \end{itemize}
Summarizing the discussion above, together with Theorem \ref{Th:3.2}, the proof is completed.
\end{proof}
\begin{table}[htbp]
\caption{Weight Distribution of $\mathcal{C}_{(p,m,k)}$}
\centering
\begin{tabular}{ll}
 \hline
 Weight& Frequency\\
 \hline
 $0$ &1\\
 $(p-1)p^{m-1}$ & $\omega+2n_{1,0,0}+2n_{1,2,0}+2n_{1,4,0}$\\
 $(p-1)p^{m-1}-p^{\frac{m-1}{2}}$ &  $(p-1)n_{(\frac{-1}{p}),0,1}+(p-1)n_{-1,1,1}$\\
 $(p-1)p^{m-1}+p^{\frac{m-1}{2}}$ &  $(p-1)n_{-(\frac{-1}{p}),0,1}+(p-1)n_{1,1,1}$\\
 $(p-1)p^{m-1}-(p-1)p^{\frac{m-1}{2}}$ &  $n_{1,1,0}$\\
 $(p-1)p^{m-1}+(p-1)p^{\frac{m-1}{2}}$ &  $n_{-1,1,0}$\\
 $(p-1)p^{m-1}-p^{\frac{m+1}{2}}$ &  $(p-1)n_{(\frac{-1}{p}),2,1}+(p-1)n_{-1,3,1}$\\
 $(p-1)p^{m-1}+p^{\frac{m+1}{2}}$ &  $(p-1)n_{-(\frac{-1}{p}),2,1}+(p-1)n_{1,3,1}$\\
 $(p-1)p^{m-1}-(p-1)p^{\frac{m+1}{2}}$ & $n_{1,3,0}$\\
  $(p-1)p^{m-1}+(p-1)p^{\frac{m+1}{2}}$ & $n_{-1,3,0}$\\
 $(p-1)p^{m-1}-p^{\frac{m+3}{2}}$ & $(p-1)n_{(\frac{-1}{p}),4,1}$\\
 $(p-1)p^{m-1}+p^{\frac{m+1}{2}}$ &$(p-1)n_{-(\frac{-1}{p}),4,1}$\\
 \hline
\end{tabular}
\end{table}
The following is an example of these codes.
\begin{example}
Let $p=3$, $m=5$ and $k=1$. The the code $\mathcal{C}_{(3,5,1)}$ is a $[242,20,81]$ cyclic code over $\mathbb{F}_{3}$ with weight enumerator
\[
\begin{split}
&1+484z^{81}+72600z^{108}+ 6853440z^{135}+84092580z^{144}+947952720z^{153}+1618713316z^{162}\\
&+782825472z^{171}+ 42810768z^{180}+3455760z^{189}+7260z^{216}.
\end{split}
\]
which is completely in agreement with the results presented in Table 3.
\end{example}
\begin{center}
    APPENDIX
\end{center}
In the following discussion, let $d_{1}(d_{2}, resp.)$ denote $p^{k}+1(p^{2k}+1, resp.)$.

\noindent $\mathbf{Proof}$ $\mathbf{of}$ $\mathbf{Lemma}$ \ref{Le:2.5}.
For any $(\overline{a},\overline{b},\overline{c}) \in \mathbb{F}_{p^{m}}^{3}$, let $N_{1(\overline{a},\overline{b},\overline{c})}$ and $N_{2(\overline{a},\overline{b},\overline{c})}$ denote the number of solutions  of the following two system of equations
\begin{equation}\label{4.16}
\begin{cases}
x_{1}^{2}+x_{2}^{2}=\overline{a}\\
x_{1}^{d_{1}}+x_{2}^{d_{1}}=\overline{b}\\
x_{1}^{d_{2}}+x_{2}^{d_{2}}=\overline{c}\\
\end{cases}
\end{equation}
\begin{equation}\label{4.17}
\begin{cases}
x_{3}^{2}+x_{4}^{2}=-\overline{a}\\
x_{3}^{d_{1}}+x_{4}^{d_{1}}=-\overline{b}\\
x_{3}^{d_{2}}+x_{4}^{d_{2}}=-\overline{c}.\\
\end{cases}
\end{equation}
Then we have
 \[N_{4}=\sum_{(\overline{a},\overline{b},\overline{c}) \in \mathbb{F}_{p^{m}}^{3}}N_{1(\overline{a},\overline{b},\overline{c})}N_{2(\overline{a},\overline{b},\overline{c})}.\]
\noindent Case 1, when $\overline{a}=0$. In this case, (\ref{4.16}) and (\ref{4.17}) have solutions if and only if $\overline{b}=\overline{c}=0$  since $-1$ is a non-square. Moreover, $N_{1(0,0,0)}=N_{2(0,0,0)}=1$.

\noindent Case 2, when $\overline{a}\neq 0$. In this case,  if $\overline{b}=0$ or $\overline{c}=0$, neither (\ref{4.16}) nor (\ref{4.17}) has solutions. So in the following, we consider the problem only when $\overline{b}\neq 0$ and $\overline{c}\neq 0$.
\begin{itemize}
  \item $\overline{a}$ is a nonzero square, $\overline{b}\neq 0$ and $\overline{c}\neq 0$.
 In this case, for any fixed $\overline{a}$, (\ref{4.16}) has the same number of  solutions as
\begin{equation}\label{4.18}
\begin{cases}
x_{1}^{2}+x_{2}^{2}=1\\
x_{1}^{d_{1}}+x_{2}^{d_{1}}=b\\
x_{1}^{d_{2}}+x_{2}^{d_{2}}=c\\
\end{cases}
\end{equation}
and  (\ref{4.17}) has the same number of  solutions as
\begin{equation}\label{4.19}
\begin{cases}
x_{3}^{2}+x_{4}^{2}=-1\\
x_{3}^{d_{1}}+x_{4}^{d_{1}}=-b\\
x_{3}^{d_{2}}+x_{4}^{d_{2}}=-c,\\
\end{cases}
\end{equation}
where $b=\overline{b}/\overline{a}^{d_{1}}$ and $c=\overline{c}/\overline{a}^{d_{2}}$. Clearly, $(b,c)$ runs through $\mathbb{F}_{p^{m}}^{*2}$ as $(\overline{b},\overline{c})$ does. According to the proof of  Lemma \ref{4.1} and \ref{4.2}, it can be easy to see that for any fixed $(b, c) \in \mathbb{F}_{p^{m}}^{*2}\setminus\{(1,1)\}$ such that (\ref{4.18}) have $2(p+1)$ solutions, then for $(-b, -c)$, (\ref{4.19}) also have $2(p+1)$ solutions. Therefore, in this case we have
\begin{equation*}
\begin{split}
&\sum_{(\overline{a},\overline{b},\overline{c}) \in \mathbb{F}_{p^{m}}^{*3}}N_{1(\overline{a},\overline{b},\overline{c})}N_{2(\overline{a},\overline{b},\overline{c})}\\
&=(p^{m}-1)\Big\{(p+1)^{2}+(2(p+1))^{2}\frac{p^{m}-p}{2(p+1)}\Big\}\\
&=(p+1)(p^{m}-1)(2p^{m}-p+1).
\end{split}
\end{equation*}
  \item $\overline{a}$ is a non-square.
 In this case, for any fixed $\overline{a}$,  (\ref{4.16}) has the same number of  solutions as
\begin{equation*}
\begin{cases}
x_{1}^{2}+x_{2}^{2}=-1\\
x_{1}^{d_{1}}+x_{2}^{d_{1}}=-b\\
x_{1}^{d_{2}}+x_{2}^{d_{2}}=-c\\
\end{cases}
\end{equation*}
and  equation system (\ref{4.17}) has the same number of  solutions as
\begin{equation*}
\begin{cases}
x_{3}^{2}+x_{4}^{2}=1\\
x_{3}^{d_{1}}+x_{4}^{d_{1}}=b\\
x_{3}^{d_{2}}+x_{4}^{d_{2}}=c.\\
\end{cases}
\end{equation*}
It can be easily seen that this case is equivalent to the case when $\overline{a}$ is a nonzero square. So  when $\overline{a}$ is a non-square, we also have
\begin{equation*}
\begin{split}
&\sum_{(\overline{a},\overline{b},\overline{c}) \in \mathbb{F}_{p^{m}}^{*3}}N_{1(\overline{a},\overline{b},\overline{c})}N_{2(\overline{a},\overline{b},\overline{c})}\\
&=(p+1)(p^{m}-1)(2p^{m}-p+1).
\end{split}
\end{equation*}
\end{itemize}
Summarizing  the two cases above, we have $N_{4}=(p+1)(p^{m}-1)(2p^{m}-p+1)+1$.
\begin{lemma}\label{4.1}
Let $N_{1(b,c)}$ denote the number of solutions $(x_{1}, x_{2}) \in \mathbb{F}_{p^{m}}^{2}$ of (\ref{4.18}), where $(b,c) \in \mathbb{F}_{p^{m}}^{*2}$. Then we have the following conclusions.
\begin{enumerate}
  \item  $N_{1(1,1)}=p+1$;\\
  \item When $(b,c)$ runs through $\mathbb{F}_{p^{m}}^{*2}\setminus\{(1,1)\}$,
\begin{equation*}
N_{1(b,c)}=
\begin{cases}
2(p+1),& for~ \frac{p^{m}-p}{2(p+1)}~ times,\\
0, & for~ the ~rest.
\end{cases}
\end{equation*}
 \end{enumerate}
\end{lemma}
\begin{proof}
We first compute the number $N_{1(b)}$ of solutions $(x_{1}, x_{2}) \in \mathbb{F}_{p^{m}}^{2}$ of the following system of equations
\begin{equation}\label{20}
\begin{cases}
x_{1}^{2}+x_{2}^{2}=1\\
x_{1}^{d_{1}}+x_{2}^{d_{1}}=b.\\
\end{cases}
\end{equation}
When $p\equiv 3 \pmod 4$, $-1$ is a non-square in $ \mathbb{F}_{p^{m}}$. However, we can
choose $t \in \mathbb{F}_{p^{2m}}$ such that $t^{2}=-1$. From the first equation of (\ref{20}), by setting $\theta =x_{1}-tx_{2} \in \mathbb{F}_{p^{2m}}^{*}$, we can have
\begin{equation}\label{21}
x_{1}=\frac{\theta + \theta^{-1}}{2},
x_{2}=\frac{t(\theta - \theta^{-1})}{2}.
\end{equation}
Substituting (\ref{21}) into the second equation of (\ref{20}), we obtain
\begin{equation}\label{22}
\theta^{p^{k}+1}+\theta^{-p^{k}-1}=2b.
\end{equation}
Denote $\theta^{p^{k}+1}$ by $w$, Eq. (\ref{22}) is equivalent to
\begin{equation}\label{23}
w^{2}-2bw+1=0.
\end{equation}
If Eq. (\ref{23}) has no solution, i.e., $b^{2}-1$ is a non-square of $\mathbb{F}_{p^{2m}}^{*}$, then $N_{b}=0$. Otherwise, let $w_{1}$ and $w_{2}=w_{1}^{-1}$ be two solutions of (\ref{23}).
Since $x_{1} \in  \mathbb{F}_{p^{m}}$, then the following holds:
\begin{equation*}
\frac{\theta+\theta^{-1}}{2}=(\frac{\theta+\theta^{-1}}{2})^{p^{m}}=\frac{\theta^{p^{m}}+\theta^{-p^{m}}}{2},
\end{equation*}
which implies $\theta^{p^{m}+1}=1$ or $\theta^{p^{m}-1}=1$.
\begin{itemize}
  \item If $\theta^{p^{m}+1}=1$, then $x_{2}^{p^{m}}=(\frac{t(\theta - \theta^{-1})}{2})^{p^{m}}=\frac{t^{p^{m}}(\theta^{p^{m}}-\theta^{-p^{m}})} {2}=\frac{t(\theta - \theta^{-1})}{2}=x_{2}$ since $t^{p^{m}}=-t$.  And then $x_{2} \in \mathbb{F}_{p^{2m}}$. According to the discussion above, we have
      \begin{equation}\label{24}
      \theta^{p^{k}+1}=w_{1}, \theta^{p^{m}+1}=1,
      \end{equation} or
      \begin{equation}\label{25}
      \theta^{p^{k}+1}=w_{1}^{-1}, \theta^{p^{m}+1}=1.
      \end{equation}
      If $\theta_{1}$ and $\theta_{2}$ are two solutions of (\ref{24}), then $(\theta_{1}/\theta_{2})^{p^{k}+1}=1=(\theta_{1}/\theta_{2})^{p^{m}+1}$. Observe that gcd$(p^{k}+1, p^{m}+1)=p+1$, then $(\theta_{1}/\theta_{2})^{p+1}=1$. So if (\ref{24}) has solutions, then it has exactly $p+1$ solutions. If $w_{1}=w_{1}^{-1}$, then (\ref{25}) is the same as (\ref{24}). In this case we have $w_{1}=\pm 1$ and then from Eq. (\ref{23}), $b=\pm 1$. But when $b=-1$, $\theta^{p^{k}+1}=w_{1}=-1$. By $\theta^{p^{m}+1}=1$ and gcd$(p^{m}+1, 2(p^{k}+1))=p+1$, we have $\theta^{p+1}=1$. And then $\theta^{p^{k}+1}=1$, which is a contradiction. So in the following we only consider $b=1$. In this case, $w_{1}=1$ and ($\ref{25}$) have $p+1$ solutions of $\theta$. Moreover, we have $p+1$ solutions of (\ref{20}).  If $w_{1}\neq w_{1}^{-1}$, then (\ref{25}) has the same number of solutions as $(\ref{24})$. Moreover, their solutions are distinct since $w_{1} \neq \pm1$. Therefore, ($\ref{24}$) and ($\ref{25}$) both have $p+1$ solutions or no solutions in $\mathbb{F}_{p^{2m}}$. By (\ref{21}), $(x_{1}, x_{2})$ is uniquely determined by $\theta$. Then (\ref{20}) has $2(p+1)$ solutions or no solutions in $\mathbb{F}_{p^{m}}^{2}$
  \item If $\theta^{p^{m}+1}\neq 1$, then $\theta^{p^{m}-1}=1$. In this case,  $\theta \in \mathbb{F}_{p^{m}}^{*}$.  Since $t \in \mathbb{F}_{p^{2m}}^{*}$, so $x_{2}=\frac{t(\theta - \theta^{-1})}{2} \in \mathbb{F}_{p^{m}}^{*}$ if and only if  $\theta= \theta^{-1}$. Then $\theta^{p^{m}+1}=\theta^{p^{m}-1}\cdot \theta^{2}= 1$ which is a contradiction.
  \end{itemize}
Until now, we have  $N_{1(1)}=p+1$ and $N_{1(b)}=0$ or $2(p+1)$ for $b\neq 1$. And as in Lemma 5.4 in \cite{20},
we define
\begin{equation*}
T=\#\{b \in \mathbb{F}_{p^{m}}: N_{1(b)}=2(p+1)\}.
\end{equation*}
Then we have
  \begin{equation*}
T=\frac{p^{m}-p}{2(p+1)}.
\end{equation*}
Similarly, we also have $c$ is uniquely determined by $b$ and  $c=1$ if and only if $b=1$.
 The proof is finished.
\end{proof}
\begin{lemma}\label{4.2}
Let $N_{2(b,c)}$ denote the number of solutions $(x_{1}, x_{2}) \in \mathbb{F}_{p^{m}}^{2}$ of (\ref{4.19}), where $(b,c) \in \mathbb{F}_{p^{m}}^{*2}$. Then we have the following conclusions.

\noindent (1). $N_{2(1,1)}=p+1$.

\noindent(2). When $(b,c)$ runs through $\mathbb{F}_{p^{m}}^{*2}\setminus\{(1,1)\}$,
\begin{equation*}
N_{2(b,c)}=
\begin{cases}
2(p+1),& for~ \frac{p^{m}-p}{2(p+1)}~ times,\\
0, & for~ the ~rest.
\end{cases}
\end{equation*}
\end{lemma}
\begin{proof}
The proof is similar to the proof of the lemma above, so we omit the details.
\end{proof}

\end{document}